\documentclass[11pt,noparskip]{article}
\usepackage[pagebackref,nosubcaption]{style}
\usepackage{thmtools}
\usepackage{thm-restate}
\usepackage{amssymb,amsmath,amsthm,amsfonts}
\usepackage{graphicx}
\usepackage[font=small]{caption}
\usepackage{float}
\usepackage[ruled,vlined,linesnumbered]{algorithm2e}
\usepackage{algpseudocode}
\usetikzlibrary{calc,decorations.pathreplacing,backgrounds}
\newcommand{\B}{\{0,1\}}

\title{Quantum circuit compilation with constant overhead}
\author{Chenyi Zhang\footnote{Google Quantum AI. Authors are in \emph{reverse} alphabetical order.} ${}^{,\!\,}$\footnote{Stanford University}
\qquad
Xinyu Tan${}^{*,}$\footnote{MIT}
\qquad
Robin Kothari${}^{*}$ \\[0.25em]
David Gosset${}^{*,}$\footnote{Institute for Quantum Computing, University of Waterloo and Perimeter Institute for Theoretical Physics}
\qquad
Craig Gidney${}^{*}$
}

\date{\vspace{-2em}}

\begin{document}

\maketitle
\begin{abstract}

Compiling a continuous gate set to a discrete universal gate set generally incurs an overhead. 
A circuit composed of $G$ arbitrary one- and two-qubit gates can be $\epsilon$-approximated by a sequence of gates from $\{H,T,\mathrm{CNOT}\}$ by replacing each original gate by a sequence of $O(\log(G/\epsilon))$ elementary gates. We  show that this overhead can be avoided using adaptivity for polynomial-sized circuits with inverse-polynomial target error, which is optimal.

\end{abstract}

%===============================================================================

\section{Introduction}\label{sec:intro}

It's often convenient to describe a quantum algorithm in the circuit model using the (infinite) gate set of \emph{arbitrary} one- and two-qubit gates. We do this with the understanding that the circuit could later be compiled down to a universal discrete gate set that is dictated by a specific choice of hardware and fault-tolerance scheme. Here we consider the discrete gate set
\begin{equation}
    \calG = \{H,T,\mathrm{CNOT}\},
\end{equation}
which is a natural instruction set for some well-known fault-tolerant architectures. Our goal in this paper is to reduce the cost of this compilation step.  

Compilation generally incurs an overhead because an arbitrary single-qubit gate cannot be represented \emph{exactly} as a product of gates from a finite gate set. The standard approach is to approximate each gate using the Solovay--Kitaev theorem or number-theoretic methods.  For example, Ross and Selinger~\cite{ross2014optimal} showed how to implement any one- or two-qubit gate to within error $\eps$ without ancillas using $O(\log(1/\epsilon))$ gates over $\calG$.
This is optimal~\cite{beverland2020lower}, even with unlimited ancillas, unlimited Clifford gates, and the ability to use an adaptive quantum circuit (i.e., one in which we can perform intermediate measurements and have future quantum gates depend on these measurement outcomes).

Now suppose we're given a circuit with $G$ arbitrary one- and two-qubit gates and our goal is to implement it within error $\epsilon$ over the gate set $\calG$. We need to ensure that sum of the errors across all $G$ gates does not exceed our budget $\epsilon$. We can therefore afford to implement each gate with error $\epsilon/G$, which costs $O(\log(G/\epsilon))$  gates over $\calG$. In total, the compiled circuit has $O(G \log(G/\epsilon))$ gates. Even if the target error is \emph{constant}, this compilation scheme leads to a total of 
\begin{equation}
    O(G \log G) \ \text{gates,}
\end{equation}
corresponding to a multiplicative overhead of $O(\log G)$ over the original circuit size. This overhead is mild, which is why it is often ignored or abstracted away. 

Our main result shows that this overhead is unnecessary when compiling to an adaptive quantum circuit. In fact, we show that with only constant overhead, we can achieve error $\epsilon=1/\poly(G)$ for any given polynomial, which is sufficient for practical applications.

\begin{theorem}
Any quantum circuit composed of $G$ arbitrary one- and two-qubit gates can be implemented to within any given $1/\mathrm{poly}(G)$ error in diamond distance by an adaptive circuit that uses
\begin{equation}
    O(G) \ \text{gates}    
\end{equation}
in the worst case from the gate set $\calG = \{H,T,\mathrm{CNOT}\}$.
\label{thm:main}
\end{theorem}

We remark that a similar result can also be proved for $T$-count (the number of $T$ gates used) in the model where Cliffords are free.
For example, a circuit with an arbitrary number of Cliffords and $G$ single-qubit $Z$-rotations can be implemented to error $1/\sqrt{G}$ using at most $8G+o(G)$ $T$ gates (see \Cref{remark:tcount2}).

\paragraph{Optimality and special cases.}
It is natural to wonder if we can do even better. Can we achieve error asymptotically smaller than $1/\poly(G)$ using only $O(G)$ gates? We show via a counting argument that this is impossible even for simple depth-1 circuits composed of $n$ single-qubit rotations. In particular, achieving a superpolynomially small error requires strictly superlinear in $n$ gates (\Cref{prop:optimal}).

This rules out better algorithms for compiling quantum circuits that are initially composed of \emph{arbitrary} one- and two-qubit gates. But it leaves the door open to improvements if the circuit is initially expressed over some other special gate set. One such gate set of practical importance is what we call the dyadic (powers of $2$) gate set, which contains $\calG$ and all rotations about $Z$ that are power-of-2 fractions of $\pi$. More precisely, define $R(\theta)=\mathrm{diag}(1,e^{i\theta})$. Then the dyadic gate set is 
\begin{equation}
\mathcal{D} = \calG \cup\{ R(\pi/2^j): j\geq 1\}.
\label{eq:dyadicgates}
\end{equation}
This gate set is sufficient to express algorithms like the quantum Fourier transform and is a natural example of an infinite gate set that appears in quantum algorithms.

What is the best we can hope for? Certainly we can't expect to compile a circuit with $G$ gates over $\mathcal D$ to fewer than $G$ gates over $\calG$. On the other hand, even if $G=1$, the $\Omega(\log(1/\eps))$ lower bound of Beverland et al.~\cite{beverland2020lower} holds even for dyadic gates (\Cref{prop:dyadic}). Thus the best case scenario is to compile a circuit with $G$ dyadic gates into a circuit with $O(G + \log(1/\eps))$ gates from $\calG$. Our second result shows that this is achievable.

\begin{theorem}\label{thm:dyadic}
Any quantum circuit composed of $G$ one- and two-qubit gates from the gate set $\mathcal D$ defined in \cref{eq:dyadicgates} can be implemented to within diamond distance $\epsilon$ by an adaptive circuit that uses
\begin{align}
    O(G+\log(1/\epsilon)) \ \text{gates}
    \label{eq:runtime_main}
\end{align}
in the worst case from the gate set
$\calG$.
\end{theorem}

The implementation described in \Cref{thm:dyadic} can also be made extremely ancilla-efficient, using at most $O(\log(G/\epsilon))$ ancillas in total. To do this, simply divide up the original $G$ gate circuit into $M\sim G/G'$ subcircuits of size $G'=\Theta(\min\{G,\log(G/\epsilon)\})$ and use the theorem to implement each subcircuit with error $\epsilon/M$. The number of ancillas needed for a given subcircuit is at most the total gate cost $O(\log(G/\epsilon))$ and these ancillas can be reused.

We note that \Cref{thm:dyadic} is novel even for a simple constant-depth circuit with only single-qubit gates. In particular, consider the \emph{phase gradient state}, defined as
\begin{equation}
|\mathrm{grad}_n\rangle=\bigotimes_{j=0}^{n-1} R(\pi/2^j)|+\rangle,   
\label{eq:pg}
\end{equation}
which can clearly be prepared using $n$ $Z$-rotations and $n$ Hadamard gates. If we set $\eps=2^{-n}$, an optimal construction of this matching \Cref{thm:dyadic} would use $O(n)$ gates over $\calG$. 

Previously, Jones~\cite{JonesFourier14} showed how to prepare this state using an adaptive quantum circuit with error $2^{-n}$ with expected gate count $O(n\log n)$. 
We show that neither the log factor, the expected runtime bound, nor the adaptivity is necessary in Jones' result; we give an optimal unitary state preparation algorithm with $O(n)$ worst-case cost.

\begin{restatable}[Optimal phase gradient state preparation]{theorem}{phasegrad}
For any $n\geq 1$, there is a unitary circuit that prepares the phase gradient state $|\mathrm{grad}_n\rangle$ to within error $2^{-n}$ using $O(n)$ gates over the set $\calG$.
\label{thm:optimalphase}
\end{restatable}

\paragraph{Proof overview of \Cref*{thm:dyadic}.} 
We start by explaining the proof of \Cref{thm:dyadic} (proved in \Cref{sec:dyadic}), which is easy to understand if we take \Cref{thm:optimalphase} for granted.

The phenomenon of mass production motivates why such low-overhead compilation results might be possible at all. 
Mass production theorems~\cite{Uhlig1974,Uhlig1992,Kretschmer2023,HKW25,GKW26} show how some tasks can be performed multiple times at the cost of one. For example, Kretschmer~\cite{Kretschmer2023} shows how to prepare superpolynomially many copies of an arbitrary $n$-qubit state with only $O(2^n)$ gates, which is the cost of preparing one copy of an arbitrary state.

With this perspective, let us analyze the problem of compiling a circuit expressed over the dyadic gate set. 
To exploit mass production, we will implement the various dyadic rotations by synthesizing some large resource state at the beginning. 
The natural resource states to prepare are $R(\pi/2^j)|+\rangle$ for every $j$ up to some maximum cutoff value dictated by the target error, since one copy of that state allows us to implement the gate $R(\pi/2^j)$ with probability $1/2$ using the standard technique of state injection. This resource state is exactly the phase gradient state, which we have assumed we can prepare optimally using \Cref{thm:optimalphase}.

However, there remains one problem with this strategy.  The state injection gadget is probabilistic and fails with probability $1/2$. When it fails, the required correction is another dyadic gate with twice the angle. Naively, we might need to stockpile a large collection of resource states to account for all possibilities, which would go over our budget.
Instead, we use a gadget we call the ``phase halving gadget'' first introduced by Gidney and Fowler~\cite{gidney2019efficient} (see \Cref{fig:rotation-resource-conversions}) that is similar to the state injection gadget, but instead of consuming a state to apply a gate, this gadget uses one single-qubit rotation and a catalyst resource state to produce two copies of the resource state corresponding to half the rotation angle. This allows us to maintain a bank of only $O(\log(1/\eps))$ resource states, which are consumed \textit{and replenished} as the algorithm proceeds, by the state injection gadget and phase halving gadget, respectively.

\paragraph{Proof overview of \Cref*{thm:main}.} 

To prove our main result for circuits composed of arbitrary one- and two-qubit gates (\Cref{thm:main}), we will again use the state injection and phase halving gadgets, as well as some additional tools. For simplicity, we first consider the special case where the target error is $\eps=1/\sqrt{G}$, which might already suffice for many applications. This special case is discussed in more detail in \Cref{sec:fixederror}.

So fix $\eps=1/\sqrt{G}$ for now. We shall assume without loss of generality that all gates in the circuit are either CNOT gates, $H$ gates,  or $Z$-rotations $R(\theta)$ with $\theta\in (0,2\pi)$ (since any one- or two-qubit gate can be expressed as a product of a constant number of such gates). We begin by approximating each of these single-qubit rotation angles by nearby angles on a uniform grid. In particular, we aim to replace each $R(\theta)$ with gates of the form $R(2\pi j/D)$ for $D=G^{3/4}$ and $j\in[D]$. At first, this appears to be too coarse an approximation, since rounding a rotation $R(\theta)$ to its closest multiple of $2\pi/D$ incurs $\sim 1/D$ error and thus the circuit has total error $\sim G/D \gg 1$. However, if we use the strategy of randomized compilation, introduced independently by Campbell~\cite{Cam17} and Hastings~\cite{Has17}, this gets us error $1/D^2$ per gate for a total error of $G/D^2= 1/\sqrt{G}$. This involves randomly choosing and then implementing a unitary corresponding to one of the two nearest grid points $R(2\pi j/D)$ or $R(2\pi(j+1)/D)$ with prescribed probabilities.

It remains to implement a circuit where all gates are $H,\mathrm{CNOT}$, or $Z$-rotations by angles that are integer multiples of $2\pi/D$. To implement the latter gates we shall use resource states just like in the dyadic case. Rather than preparing two copies of the phase gradient state, here we prepare two copies of each state $R(2\pi j/D)\ket{+}$ for all $j \in [D]$. In this case, since we only need inverse polynomial error, these $O(G^{3/4})$ states can be prepared straightforwardly using the Ross--Selinger algorithm at total cost $O(G^{3/4}\log G)$. Combining this with the state injection and phase halving gadgets yields our implementation protocol for the special case $\eps=1/\sqrt{G}$.

To handle the more general case (\Cref{thm:main}) in which $\epsilon$ may be any inverse polynomial function of $G$, we need one more idea. Say for concreteness our target error is $\eps=O(1/G^{9})$. 
Then it is certainly sufficient to discretize all rotations to multiples of $2\pi/G^{10}$ since the overall error will be $O(G/G^{10})$. But to reuse the previous strategy, we would need to prepare a state on $G^{10}$ qubits, which is far too expensive. 
Instead, we use an encoding trick to reduce the size of the resource state.
We can represent any integer in $[G^{10}]$ in base $B$, where $B=\sqrt{G}$. 
Since this is a large base, every number in $\{0,\ldots,G^{10}-1\}$ is merely a 20-digit number. Thus any rotation by an angle that is a multiple of $2\pi/G^{10}$ can be written as a product of $20$ gates, each of which is now from a set of size $B=\sqrt{G}$. This idea allows us to prepare a smaller resource state and achieve the main result. 
This also clarifies why we get gate count $O(G)$ only for a fixed polynomial like $G^{10}$, since the exponent of the polynomial appears multiplicatively in the overhead.

\paragraph{Proof overview of \Cref*{thm:optimalphase}.}

It will be convenient to abuse notation slightly and write
the
$n$-qubit phase gradient state as
\begin{equation}
|\mathrm{grad}_n\rangle\equiv e^{-\ii H}|+^n\rangle\qquad  \qquad H=\pi\sum_{j=1}^{n} 2^{-j}Z_j.
\end{equation}
Up to a global phase, this agrees with our definition in~\cref{eq:pg}.

The main idea is to write
\begin{equation}
e^{-\ii H}=\sum_{k=0}^{\infty} \frac{1}{k!}\sum_{\ell_1,\ell_2,\ldots \ell_k=1}^{n} \prod_{j=1}^{k} (-\ii\pi 2^{-\ell_j} Z_{\ell_j})
\end{equation}
where the RHS is a linear combination of unitaries, and then to use the LCU method \cite{ChildsWiebe12,BCCKS15} to implement it. One way to see that this might be a promising approach is that the $1$-norm of the coefficients appearing on the RHS, which in many applications dominates the cost of LCU state preparation, is at most
\begin{equation}
\sum_{k=0}^{\infty}\frac{1}{k!}\left(\sum_{\ell=1}^{\infty}(\pi/2^\ell)\right)^k=e^{\pi}
\end{equation}
which is a constant independent of $n$. So in this case there is no large overhead arising from the 1-norm. This is good news but obtaining a linear time algorithm is nevertheless delicate; it requires a $O(n)$-runtime implementation for each subroutine needed for the LCU technique, see \Cref{sec:phasegrad} for details. For instance, one challenge is to implement a ``SELECT'' operation
\begin{equation}
|(\ell_1,\ell_2,\ldots, \ell_k)\rangle|\phi\rangle\rightarrow |(\ell_1,\ell_2,\ldots, \ell_k)\rangle\prod_{j=1}^{k} Z_{\ell_j}|\phi\rangle.
\end{equation}
which takes as input a (suitably encoded) tuple of positive integers $(\ell_1,\ell_2,\ldots, \ell_k)\in [n]^{k}$ such that $\sum_{j}\ell_j\leq O(n)$ and applies a corresponding Z-type Pauli to a given $n$-qubit state $|\phi\rangle$ (note that without the restriction on $\sum_{j}\ell_j$, the number of input data bits could exceed our runtime budget). Our implementation of this SELECT operation works with a tailored encoding of the input $(\ell_1,\ell_2,\ldots, \ell_k)$ and makes use of recently developed fast classical circuits for multiselection due to Holmgren and Rothblum \cite{HolmgrenRothblum24}.

\paragraph{Discussion.}
We have shown that circuits composed of one- and two-qubit gates can be compiled into adaptive circuits over a discrete gate set with only a constant overhead. The protocol that achieves this is also somewhat practical--- the ingredients we need are just simple circuit gadgets and an adaptive strategy for managing certain single-qubit resource states. Moreover, as we discussed above (and detail in \Cref{sec:fixederror}), the protocol is even easier to implement if our target diamond-distance error is at most $1/\sqrt{G}$, which is suitable for most practical applications. We also showed that even better scaling with error $\epsilon$ is possible for compilation of circuits containing dyadic gates (\Cref{thm:dyadic}), though we caution that the optimal phase gradient state preparation (\Cref{thm:optimalphase}) used in the protocol requires complex reversible circuits to be implemented coherently and is not practical in the near term.   However, we note that it is possible to directly substitute Jones's phase gradient state algorithm~\cite{JonesFourier14} in our protocol (described in \Cref{sec:dyadic}). That would give a nearly optimal and significantly more practical expected-cost version of \Cref{thm:dyadic} with a slightly worse gate count than the one we obtain, namely $O(G + \log(1/\eps)\log\log(1/\eps))$. 

A natural open question that we leave for future work is whether this kind of constant-overhead compilation can be achieved without the use of full adaptivity, for example via a mixed unitary channel implemented by randomly choosing a circuit from some distribution and then implementing it on a quantum computer.

The rest of this paper is organized as follows. In \Cref{sec:dyadic} we prove \Cref{thm:dyadic} assuming \Cref{thm:optimalphase}. In \Cref{sec:fixederror} we describe a protocol that proves \Cref{thm:main} in the special case $\epsilon=1/\sqrt{G}$. Then in \Cref{sec:mainresult} we prove \Cref{thm:main} in the general case. Finally, in \Cref{sec:phasegrad} we prove \Cref{thm:optimalphase}.

%===============================================================================

\section{Circuits with dyadic gates}
\label{sec:dyadic}

In this section we prove \Cref{thm:dyadic}.

First, we delete gates in the original circuit with tiny rotation angles that can safely be neglected within our error budget. Let $K=\lceil\log_2(2\pi G/\epsilon)\rceil$. Deleting all rotation gates $R(\pi/2^j)$ with $j>K$
incurs at most an overall error of
\begin{equation}
G\cdot \pi/2^{K+1}\leq \epsilon/4
\end{equation}
in diamond norm distance. 
In the remainder we assume that all rotations are of the form $R(\pi/2^j)$ with $0\leq j\leq K$.

To implement these gates we will use two quantum circuit gadgets shown in \Cref{fig:rotation-resource-conversions}, coupled with a novel strategy for stockpiling and deploying the resource states used within these gadgets.

\begin{figure}[htbp!]
\centering

\begin{minipage}[c]{0.39\textwidth}
\centering

\resizebox{0.98\linewidth}{!}{
\begin{quantikz}[
  row sep=0.55cm,
  column sep=0.45cm
]
\lstick{$\lvert\psi\rangle$}
  & \ctrl{1}
  & \qw
  & \gate{R(2\theta)}
  & \qw
    \rstick{$R(\theta)\lvert\psi\rangle$}
\\
\lstick{$R(\theta)\lvert+\rangle$}
  & \targ{}
  & \meter{}
  & \control{}\setwiretype{c}\vcw{-1}
\end{quantikz}
}

\vspace{1.0em}

\end{minipage}%
\hspace{0.015\textwidth}%
\begin{minipage}[c]{0.012\textwidth}
\centering
\begin{tikzpicture}
  \draw[
    densely dotted,
    line width=0.55pt
  ] (0,0) -- (0,3cm);
\end{tikzpicture}
\end{minipage}%
\hspace{0.015\textwidth}%
\begin{minipage}[c]{0.555\textwidth}
\centering

\resizebox{\linewidth}{!}{
\begin{quantikz}[
  row sep=0.50cm,
  column sep=0.24cm
]
\lstick{$R(\theta)\lvert+\rangle$}
  & \targ{}
  & \targ{}
  & \ctrl{1}
  & \ctrl{2}
  & \qw
  & \ctrl{3}
  & \qw
  & \ctrl{3}
  & \qw
  & \ctrl{2}
  & \ctrl{1}
  & \rstick{$R(\theta)\lvert+\rangle$}
\\
\lstick{$\lvert\psi\rangle$}
  & \ctrl{-1}
  & \qw
  & \targ{}
  & \qw
  & \ctrl{2}
  & \qw
  & \qw
  & \qw
  & \ctrl{2}
  & \qw
  & \targ{}
  & \rstick{$R(\theta)\lvert\psi\rangle$}
\\
\lstick{$\lvert\phi\rangle$}
  & \qw
  & \ctrl{-2}
  & \qw
  & \targ{}
  & \ctrl{1}
  & \qw
  & \qw
  & \qw
  & \ctrl{1}
  & \targ{}
  & \qw
  & \rstick{$R(\theta)\lvert\phi\rangle$}
\\
\lstick{$\lvert0\rangle$}
  & \qw
  & \qw
  & \qw
  & \qw
  & \targ{}
  & \targ{}
  & \gate{R(2\theta)}
  & \targ{}
  & \targ{}
  & \qw
  & \qw
  & \qw
  \rstick{$\ket{0}$}
\end{quantikz}
}

\vspace{1.0em}

\end{minipage}

\caption{State injection (left) and phase-halving (right) gadgets }
\label{fig:rotation-resource-conversions}
\end{figure}
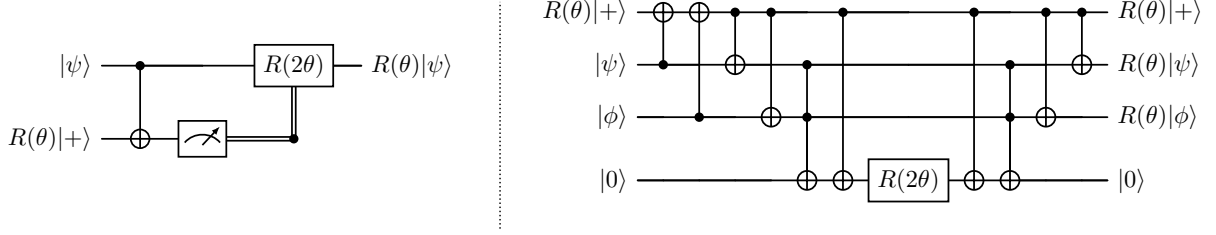

The circuit shown on the left is the familiar \textit{state injection gadget} which adaptively implements the unitary $R(\theta)$. It consumes one copy of the state $R(\theta)|+\rangle$ and with probability $1/2$ (classically controlled on the measurement outcome) also uses an $R(2\theta)$ gate. It is easy to check that this gadget has the claimed input/output behavior. The corresponding resource conversion is summarized as follows in expectation:
\begin{equation}
R(\theta)\lvert+\rangle
  +\frac{1}{2}R(2\theta)
  \xRightarrow{}
  R(\theta).
  \label{eq:res1}
\end{equation}

The circuit shown on the right is a \textit{phase-halving gadget}, first introduced in \cite{gidney2019efficient}. It requires a resource state $R(\theta)|+\rangle$, which is a catalyst that is returned in the same state at the end of the circuit. This gadget applies the unitary $R(\theta)$ twice, to each of two input states, and consumes the unitary $R(2\theta)$. 
We analyze this gadget in more detail and establish this input/output behavior in \Cref{prop:phasehalving}. Its functionality is summarized as follows:
\begin{equation}
R(2\theta)
\xRightarrow{\,R(\theta)\lvert+\rangle\,}2R(\theta).
\label{eq:res2}
\end{equation}

Both of these gadgets make use of resource states, so we will need to have some of those around. A collection of resource states that contains one copy of each dyadic rotation angle (up to some maximum) is exactly the phase gradient state defined in \cref{eq:pg}:
$|\mathrm{grad}_K\rangle=\bigotimes_{j=0}^{K-1} R(\pi/2^j)|+\rangle$.

We will initially prepare two copies of the $(K+1)$-qubit phase gradient state. 
For each $0\leq j\leq K$ we then have two copies
\begin{equation}
R(\pi/2^j)|+\rangle \otimes R(\pi/2^j)|+\rangle
\label{eq:twocopies}
\end{equation}
We call one copy the \textit{catalyst} and the other copy the \textit{bank}. The catalyst copy will always be present throughout the course of implementing the circuit and is only ever used in the phase-halving gadget (which returns it unharmed). On the other hand, the bank copy can be withdrawn during the algorithm when we use the state injection gadget (that consumes it), and later deposited back at some later point when we use the phase-halving gadget to generate a new copy. In particular, our strategy for implementing each single-qubit dyadic rotation $R(\theta)$ in the circuit is described by \Cref{alg:recursive_rotation}.

\begin{algorithm2e}[htbp!]
\caption{Implement a single-qubit rotation}
\label{alg:recursive_rotation}

\KwIn{An angle $\theta$ specifying the gate $R(\theta)$ to be implemented.}

\If{$R(\theta)$ is the identity gate }{
    Do nothing.
}
\Else{
    \If{there are two copies of $R(\theta)\ket{+}$}{
        Withdraw the bank copy and use it to implement $R(\theta)$ via the state injection gadget. If the correction $R(2\theta)$ is needed, implement it via \Cref{alg:recursive_rotation}.
    }
    \ElseIf{there is one copy of $R(\theta)\ket{+}$}{
        Perform the desired gate $R(\theta)$ and at the same time prepare a new copy of $R(\theta)\ket{+}$ to be deposited into the bank,  using the phase-halving gadget. This requires the gate $R(2\theta)$ which we implement via \Cref{alg:recursive_rotation}.
    }
}
\end{algorithm2e}

Suppose our implementation of the entire circuit recursively induces $N$ calls to \Cref{alg:recursive_rotation} in total. 
Instead of waiting for all $N$ calls to finish, we can terminate after some fixed number of $t_{\mathrm{max}}$ calls. 
Let $p = \Pr[N > t_{\max}]$. Then with probability $1-p$, we end up implementing the ideal circuit and with probability $p$ we implement something else (which could possibly be far from the ideal circuit). This incurs an error of at most $2p$ in diamond norm distance. 

Using \Cref{thm:optimalphase}, we can prepare two copies of the phase gradient state to error $\epsilon/4$ using $O(K) = O(\log(G/\epsilon))$ gates, which translates to an error of $\epsilon/4$ in diamond norm distance for using the imperfect phase gradient states in \Cref{alg:recursive_rotation}. 

Therefore, by the triangle inequality of the diamond norm distance, the overall error of the process comes only from the error in truncating tiny rotation angles (at most $\epsilon/4$), the error in preparing the resource states (at most $\epsilon/4$),  and the error (at most $2p$) from imposing the cutoff. 

Below we show that $p= \Pr[N > t_{\max}] \leq \epsilon / 4$ when the cutoff value is given by $t_{\mathrm{max}}=O(G+\log(1/\epsilon))$. Each of the $N$ calls to \Cref{alg:recursive_rotation} uses $O(1)$ elementary gates from our gate set $\mathcal{G}$. So the total gate cost is
\begin{equation}
O(K)+O(t_{\max})=O(G+\log(1/\epsilon)),
\end{equation}
as claimed, where the first term accounts for the phase gradient state preparation. 

Let $I$ and $H$ denote the total numbers of state-injection and phase-halving calls, respectively. Hence $N = I + H$. 
For each dyadic angle $\theta=\pi/2^j$, the calls requesting $R(\theta)$ alternate, starting with the state injection because every bank is initially full. Thus $H\leq I$. 

Each state injection terminates the current recursion (i.e., no correction is needed) with probability $1/2$. 
For each $i\in [I]$, let $X_i\in \{0,1\}$ be an independent fair bit where $X_i=1$ means no correction is required at the $i$-th state-injection call. 
Since there are at most $G$ original recursion chains, then if $I > m$, we must have $\sum_{i=1}^m X_i < G$. 
Therefore, for any integer $m\geq 4G$, Hoeffding's inequality gives
\begin{equation}
    \Pr[I > m] \leq \Pr[\mathrm{Bin}(m, 1/2) < G] \leq e^{-m/8}. 
\end{equation}
Let us choose $m$ such that $e^{-m/8}\leq \epsilon/4$ and $m\geq 4G$, i.e., $m = \max \{4G, \lceil 8\ln(4/\epsilon) \rceil \}$. Set $t_{\max} = 2m$. 
Then
\begin{equation}
    \Pr[N > t_{\max}] \leq \Pr[I > m] \leq \epsilon/4, 
\end{equation}
where $t_{\max} = O(G + \log(1/\epsilon))$ as claimed. 

\begin{remark}[$T$-count]
\label{remark:tcount}
The above algorithm also uses $O(G + \log(1/\epsilon))$ $T$ gates. Furthermore, we show below that when $\log(1/\eps) = o(G)$, i.e., when we do not need exponentially small error, the circuit has $T$ count $\leq 8G + o(G)$. In particular, each $Z$-rotation only costs roughly 8 $T$ gates.

To see this, observe that the phase gradient state uses $O(\log(G/\eps)) = o(G)$ total gates by assumption.
Then the only non-Clifford resources required in this protocol appear in the phase-halving gadget. We show in \Cref{sec:phasehalving} that each phase-halving gadget can be implemented using 4 $T$ gates.  
So it suffices to show that $\E[H] \leq 2G$. 
Since there are at most $G$ original recursion chains and each state injection independently terminates the current recursion with probability $1/2$, then $\E[I] \leq 2G$. Since $H\leq I$, we have $\E[H] \leq 2G$ and thus each single-qubit dyadic $Z$-rotation costs $8$ $T$ gates in expectation. 
This can be made worst-case using a sharper analysis with Hoeffding's inequality.
\end{remark}

\section{Circuits with general gates and inverse square-root error}
\label{sec:fixederror}

In this section we describe a simple proof of \Cref{thm:main} for the special case where the target error is $\epsilon=1/\sqrt{G}$.

Our strategy involves first discretizing the set of possible non-Clifford gates in the circuit to a uniform grid. To this end, let $D$ be a positive integer and consider
\begin{equation}
\mathcal{S}_D=\{R(2\pi j/D): 0\leq j\leq D-1\}.
\label{eq:grid}
\end{equation}

If we have a single-qubit $Z$ rotation $R(\theta)$ for some $\theta\in [0,2\pi]$ we can always round $\theta$ to the nearest grid point $2\pi j/D$. In this way we can approximate $R(\theta)$ to within error $O(1/D)$ by a rotation in $\mathcal{S}_D$. Surprisingly, this is not the best we can do: if we approximate $R(\theta)$ by a distribution over gates in $\mathcal{S}_
D$ then we can achieve a better $O(1/D^2)$ scaling \cite{Cam17,Has17}.   

This \textit{randomized rounding} strategy works as follows. First write 
\begin{equation}
\theta=(j+p)2\pi/D
\end{equation}
for some integer $j$ and $0\leq p<1$. We approximate $R(\theta)$ by a randomly chosen $R(\phi)\in \mathcal{S}_D$ such that 
\begin{equation}
\phi=\begin{cases}2\pi j/D & \text{ with probability } 1-p\\
2\pi(j+1)/D & \text{ with probability } p.
\end{cases}
\end{equation}
We claim that the corresponding channel is $O(1/D^2)$-close to the desired unitary $R(\theta)$.
\begin{claim}
\begin{equation}\label{eq:randomized_rounding}
\left\|
\mathbb E_{\phi}\left[
R(\phi)(\cdot)R(\phi)^\dagger
\right]
-
R(\theta)(\cdot)R(\theta)^\dagger
\right\|_\diamond
=O(D^{-2}).
\end{equation}
\end{claim}
\begin{proof}
Let $\phi_0 = 2\pi j/D$ and $\phi_1 = 2\pi(j+1)/D$. 
We specialize Campbell's mixing lemma~\cite[Lemma~1]{Cam17}, which states that if $\norm{R(\phi_0) - R(\theta)}\leq a$, $\norm{R(\phi_1) - R(\theta)}\leq a$, and $\norm{(1-p)R(\phi_0) + pR(\phi_1)  - R(\theta)}\leq b$, then the left-hand side of \cref{eq:randomized_rounding} is at most $a^2 + 2b$. 
In our case, direct computations yield $a\leq 2\pi/D$ and $b\leq (2\pi/D)^2$, which proves the claim. 

Alternatively, the claim also follows from a direct specialization of
Hastings~\cite[Lemma~1 and Eqs.~(8)--(10)]{Has17}.
\end{proof}

\begin{theorem}
Any quantum circuit composed of $G$ arbitrary one- and two-qubit gates can be implemented to within  $\epsilon=1/\sqrt{G}$ error in diamond distance by an adaptive circuit that uses
\begin{equation}
    O(G) \ \text{gates}    
\end{equation}
in the worst case from the gate set $\calG = \{H,T,\mathrm{CNOT}\}$.
\end{theorem}

\begin{proof}
Let $D$ be a power of two that we will fix shortly. We first decompose every two-qubit gate into a constant number of CNOTs and single-qubit gates. Then we decompose each single-qubit gate using Euler angles into three rotations $R(\theta)$ along with Hadamard gates. Then, we implement each rotation $R(\theta)$ using the randomized rounding strategy described above. At this point, we have (randomly) selected a circuit composed of $O(G)$ gates from the set $\mathcal{G}\cup\mathcal{S}_D$. This implements a channel that is $\delta$-close to $C$, where
\begin{equation}
\delta=O(GD^{-2})  
\end{equation}
We choose $D=\Theta(G^{3/4})$ so that the RHS is at most $1/(4\sqrt{G})=\epsilon/4$.

In the following, we show how to implement the chosen circuit to within $\epsilon/2$ error using $O(G)$ gates from $\calG$. Let $\alpha_k=2\pi k/D$ for each $1\leq k\leq D$. We begin by preparing two resource states
\begin{equation}
R(\alpha_k)|+\rangle\otimes R(\alpha_k)|+\rangle
\end{equation}
for each $k\in [D]$. Here we require each resource state to be prepared with error $\epsilon/8D=\Theta(G^{-5/4})$ so the total error in the tensor product of all $2D$ states is at most $\epsilon/4$. We use the Ross--Selinger compilation method \cite{ross2014optimal} to prepare these states, which uses
\begin{equation}
O(D\log(D/\epsilon))=O(G^{3/4}\log(G))=o(G)
\end{equation}
gates.

Now we can use exactly the same strategy as in the previous section to implement the circuit. In particular, we implement Clifford gates directly, and anytime we need to implement a gate $R(\alpha_k)\in \mathcal{S}_D$ we call \cref{alg:recursive_rotation} with $\theta=\alpha_k$. The analysis of the runtime is then identical to that of the previous theorem, and shows that the process can be terminated after at most 
\begin{equation}
t_{\mathrm{max}}=O(G+\log(1/\epsilon))=O(G)
\end{equation}
calls to \Cref{alg:recursive_rotation}, incurring an error at most $\epsilon/2$. 

The total number of gates used to implement the circuit is then $O(G)$ in addition to the cost of preparing the initial resource states, which we have already shown is $o(G)$.
\end{proof}

\begin{remark}[$T$-count]\label{remark:tcount2}
The same analysis as in \Cref{remark:tcount} also establishes that a circuit with an arbitrary number of Cliffords and $G$ single-qubit $Z$-rotations can be implemented to error $1/\sqrt{G}$ using at most $8G+o(G)$ $T$ gates in the worst case.
\end{remark}

\section{Circuits with general gates and inverse polynomial error\label{sec:mainresult}} 
In this section we prove \Cref{thm:main}. First, note that we can decompose the one- and two-qubit gates into a constant number of gates from $\calG$ and $Z$ rotations, as before. So we may assume (with only a constant-factor increase in $G$) that the circuit is composed of gates from the set $\mathcal{G}\cup \{R(\theta): \theta\in (0,2\pi)\}$. Next, we approximate each rotation angle by the nearest point on a uniformly spaced grid. In particular, each $R(\theta)$ appearing in the circuit can be approximated by the nearest rotation in $\mathcal{S}_D$ in~\cref{eq:grid}, incurring an error which is at most $O(1/D)$ per gate, or $O(G/D)$ for the whole circuit. Choosing $D=\Omega (G/\epsilon)$ then suffices to make the error in this approximation at most $\epsilon/2$. It remains to establish \Cref{thm:main} in the special case where all gates are in $\mathcal{S}_D$ for $D=\Theta(G/\epsilon)$.

We will use one more preprocessing step. Since $\epsilon=1/\mathrm{poly}(G)$, there is a constant $r\geq 2$ such that 
\begin{equation}
(G/\epsilon)^{1/r}=O(\sqrt{G}).
\end{equation}
Fixing this choice of $r$, we choose our grid spacing such that $D=B^r$ for some positive integer $B$. In the proof below we will write each grid point $2\pi j/D$ using its base-$B$ expansion as a number of the form
\begin{equation}
\frac{2\pi}{D}\sum_{\ell=0}^{r-1} x_\ell B^\ell \qquad x_0,\ldots x_{r-1}\in \{0,1,\ldots, B-1\}.
\label{eq:basebexpand}
\end{equation}
The following Theorem is then sufficient to complete the proof of \Cref{thm:main}. Indeed, with our choices of $r,B$ described above, the stated gate count in \cref{eq:expectedrb} is at most $O(G)$ since $r=O(1), \epsilon=1/\mathrm{poly}(G)$ and
\begin{equation}
B=D^{1/r}=\left(O(G/\epsilon)\right)^{1/r}=O(\sqrt{G}).
\end{equation}

\begin{theorem}\label{thm:polynomial_grids}
Let $C$ be a quantum circuit composed of $G$ gates from the set $\mathcal{G}\cup\mathcal{S}_D$. Suppose $D,B$ are both powers of two such that $D=B^r$ for $r\geq 2$. There is an adaptive circuit that implements $C$ to diamond distance $\epsilon$ using 
\begin{align}
O(rG+rB\log(B)\log(\epsilon^{-1}Br\log(B)))
\label{eq:expectedrb}
\end{align}
gates from the set $\calG$ in the worst case.
\end{theorem}

\begin{proof}
We present an explicit adaptive process that implements $C$.
Consider the base-$B$ expansion of an angle $\theta=2\pi j/D$ as in \cref{eq:basebexpand}. Each term in the sum is one of the following numbers
\begin{equation}
2\pi B^\ell k/D \qquad 0\leq \ell\leq r-1 \qquad 0\leq k\leq B-1.
\label{eq:nums}
\end{equation}
Let $D=2^m$ so $m = O(r\log(B))$. Since $B$ is
also a power of two, each of the numbers in \cref{eq:nums} is of the form
\begin{equation}
\alpha_{k,s}=2\pi k/2^s \qquad 0\leq k\leq B-1 \qquad 0\leq s\leq m.
\end{equation}

Any gate in $\mathcal{S}_D$ can therefore be implemented by a sequence of at most $r$ gates of the form $R(\alpha_{k,s})$. To implement these gates we shall use a strategy similar to the one from the previous sections. 

We first prepare two copies of the resource state 
\begin{align}\label{eq:base_B_grid_rotation_resource}
    \bigotimes_{s=1}^{m}\bigotimes_{k=0}^{B-1}R(\alpha_{k,s})\ket{+}.
\end{align}
We want the total error in preparing both copies to be at most $\epsilon/2$, so we prepare each single qubit state to within error $\epsilon/(4Bm)$. We use Ross--Selinger synthesis \cite{ross2014optimal} which costs
\begin{equation}
O(mB\log(mB/\epsilon))=O(rB\log(B)\log(\epsilon^{-1}Br\log(B))).
\end{equation}

Once we have prepared our two copies of \cref{eq:base_B_grid_rotation_resource}, we then have two resource states $R(\alpha_{k,s})\ket{+}\otimes R(\alpha_{k,s})\ket{+}$ for every $\alpha_{k,s}$. As before, we think of one of the copies as a bank state and one as a catalyst state.  Whenever we need to apply a rotation $R(\theta)\in \mathcal{S}_D$, we write $\theta$ as a sum of $r$ angles $\alpha_{k,s}$ and use \Cref{alg:recursive_rotation} to implement each of them. The doubled angle rotation satisfies $R(2\alpha_{k,s})=R(\alpha_{k,s-1})$ and may be implemented recursively whenever $s>1$. If $s=1$ the doubled angle rotation is the identity $R(\alpha_{k,0})=I$.

Again we use the same worst-case analysis as in the proof of \Cref{thm:dyadic}. The only difference is that here we have an initial number $rG$ of calls to \Cref{alg:recursive_rotation} that must be processed. The cutoff value $t_{\mathrm{max}}=O(rG+\log(1/\epsilon))$ is then sufficient to ensure that this truncation incurs only $\epsilon/2$ error. The total number of gates used is then
\begin{equation}
O(t_{\mathrm{max}}+rB\log(B)\log(\epsilon^{-1}Br\log(B))),
\end{equation}
as claimed.
\end{proof}

\section{Optimal phase gradient state preparation}
\label{sec:phasegrad}

In this section we prove \Cref{thm:optimalphase}. Below we use the notation $\approx_{\epsilon}$ to denote that the LHS and RHS of an equation are $\epsilon$-close. In addition, we shall abuse notation slightly, writing the
$n$-qubit phase gradient state as
\begin{equation}
|\mathrm{grad}_n\rangle\equiv e^{-\ii H}|+^n\rangle\qquad  \qquad H=\pi\sum_{j=1}^{n} 2^{-j}Z_j.
\label{eq:grad1}
\end{equation}
Up to a global phase, this agrees with our definition in~\cref{eq:pg}. 

To prove \Cref{thm:optimalphase} we first establish that the phase gradient state can be approximately prepared with constant known failure probability in a heralded manner.
\begin{theorem}[Heralded approximate state preparation]
\label{thm:herald}
For any $n\geq 1$ there is a unitary circuit $W$ composed of $O(n)$ gates in total from the set $\calG$, and $a=O(n)$ ancillas, such that
\begin{equation}
(I\otimes\langle0^a|) W|0^{n+a}\rangle=\gamma|\mathrm{grad}_n\rangle+|E\rangle,
\end{equation}
where $\gamma$ is an explicit complex number defined in \cref{eq:gamma_n} satisfying $|\gamma|\geq 0.01$, and $|E\rangle$ is an $n$-qubit state satisfying $\||E\rangle\|\leq 2^{-2n}$. 
\end{theorem}
\Cref{thm:optimalphase} then follows by combining the above with exact amplitude amplification as we now describe.

\begin{proof}[Proof of \Cref{thm:optimalphase}]
Let $n\geq 1$ be fixed and let $W,\gamma,a,|E\rangle$ be as in \Cref{thm:herald}. Let $r=100$ so that
\begin{equation}
b\equiv \sin(\pi/(4r+2))
\label{eq:br}
\end{equation}
satisfies $0.0075\leq b\leq |\gamma|$, where $\gamma$ is the constant from \Cref{thm:herald}. Let $q=b^2/|\gamma|^2$ and let $U$ be a single-qubit unitary such that 
\begin{equation}
U|0\rangle \approx_{2^{-2n}}\sqrt{q}|0\rangle+\sqrt{1-q}|1\rangle.
\end{equation}
Such a unitary can be implemented via the techniques of Ref.~\cite{ross2014optimal} with $O(n)$  gates from the set $\{H,T\}$. Now let $W'=W\otimes U$ so that 
\begin{equation}
I\otimes \langle 0^{a+1}|W'|0^{n+a+1}\rangle \approx_{\delta} b\cdot \frac{\gamma}{|\gamma|}|\mathrm{grad}_n\rangle.
\end{equation}
where $\delta\leq 2^{-2n+1}$. Next we use exactly $r=100$ rounds of amplitude amplification. If $\delta$ was zero this would exactly prepare the state we are interested in due to our choice~\cref{eq:br}. Because $\delta$ is nonzero, this results in a state which is $\epsilon=O(1)\cdot 2^{-2n+1}$-close to the desired state
\begin{equation}
\frac{\gamma}{|\gamma|}|\mathrm{grad}_n\rangle.
\end{equation}
We may take $n$ sufficiently large such that $\epsilon\leq 2^{-n}$ as claimed (the theorem holds trivially for $n$ below any fixed constant threshold value).
This achieves the claimed runtime because in each round of amplitude amplification we use $O(n)$ gates from the set $\calG$. 
\end{proof}

In the remainder of this section we prove \Cref{thm:herald}. As mentioned in the introduction, we write
\begin{equation}\label{eq:exp_LCU_decomposition}
e^{-\ii H}=\sum_{k=0}^{\infty} \frac{1}{k!}\sum_{\ell_1,\ell_2,\ldots \ell_k=1}^{n} \prod_{j=1}^{k} (-\ii\pi 2^{-\ell_j} Z_{\ell_j})
\end{equation}
where the RHS is a linear combination of unitaries, and then to use the LCU method \cite{ChildsWiebe12,BCCKS15} to implement it. We shall use the following binary encoding of tuples $(\ell_1,\ell_2\ldots, \ell_k)$.

Write $\mathbb{N}_+=\{1,2,3\ldots\}$ for the positive integers. 
For any $k\geq 1$ and $(\ell_1,\ell_2,\ldots, \ell_k)\in \mathbb{N}_{+}^{k}$ define 
\begin{equation}
|(\ell_1,\ell_2,\ldots, \ell_k)\rangle\equiv|0^{\ell_1-1}10^{\ell_2-1}1\ldots 0^{\ell_k-1}1\rangle.
\end{equation}
This is a computational basis state of $\sum_{j=1}^{k}\ell_j$ qubits. For $M\in \mathbb{N}_{+}$ and $1\leq k\leq M$, define
\begin{equation}
|\phi_{k,M}\rangle\equiv \sum_{\stackrel{\ell_1,\ell_2\ldots, \ell_k\in \mathbb{N}_+:}{\sum{\ell_j}\leq M}}\frac{1}{\sqrt{2^{\sum_{j}\ell_j}}} |(\ell_1,\ell_2,\ldots, \ell_k)\rangle |+^{M-\sum_{j} \ell_j}\rangle .
\end{equation}
We also define $|\phi_{0,M}\rangle=|+^M\rangle$.
Note that $|\phi_{k,M}\rangle$ is a uniform superposition of strings of Hamming weight at least $k$, and is subnormalized, i.e.  $\||\phi_{k,M}\rangle\|\leq 1$.  
For small values of $k$, these states are well approximated by the uniform superposition $|+^M\rangle$ of all computational basis states.
\begin{lemma}
For each $0\leq k\leq M/4$ we have $\big\||+^{M}\rangle-|\phi_{k,M}\rangle \big\|\leq 2^{-M/12}$. 
\label{lem:plus}
\end{lemma}
\begin{proof}
For $k=0$ the two states are identical, while for $1\leq k\leq M/4$ we have 
\begin{equation}
\norm*{|+^{M}\rangle-|\phi_{k,M}\rangle}^2
=\norm[\bigg]{\frac{1}{\sqrt{2^{M}}}\sum_{z: |z|<k} |z\rangle}^2
=\frac{1}{2^M}\sum_{j=0}^{k-1} \binom{M}{j}\leq 2^{M\cdot h(k/M)-M}
\end{equation}
where $h(\cdot)$ is the binary entropy. Plugging in $h(k/M)\leq h(1/4)< 5/6$ completes the proof. 
\end{proof}

\begin{proof}[Proof of \Cref{thm:herald}]

Fix $n\geq1$. Let $M=128n$ and let
\begin{equation}
K=\left\lceil \frac{c n}{\log_2(n+2)}\right\rceil,
\qquad
\Lambda_K=\sum_{k=0}^{K}\frac{\pi^k}{k!},
\qquad
\alpha_k=\frac{\pi^k}{\Lambda_K k!},
\label{eq:kc}
\end{equation}
where $c$ is a sufficiently large absolute constant. Then, \cref{eq:exp_LCU_decomposition} can be expressed as
\begin{equation}
e^{-iH}= \Lambda_K\sum_{k=0}^{K}\alpha_k
\sum_{\ell_1,\ldots,\ell_k=1}^{n}
\prod_{j=1}^{k}\bigl(-i2^{-\ell_j}Z_{\ell_j}\bigr)
+ \sum_{k>K}\frac{(-iH)^k}{k!}.
\end{equation}
WLOG we assume that $n$ is sufficiently large such that 
\begin{equation}
K\leq M/4.
\end{equation}
The first step of our algorithm is to (approximately) prepare a state with coefficients $\sqrt{\alpha}_k$. The following lemma is proved in \Cref{sec:prepare}.

\begin{lemma}[PREPARE lemma]\label{lem:poisson-prepare}
Suppose $c$ in~\cref{eq:kc} is a sufficiently large constant. There is a unitary $\mathrm{PREP}_n$ using $O(n)$ gates from the set $\calG$ and $B=O(n)$  ancillas such that, for some normalized states $|g_k\rangle$,
\begin{equation}
\left\|
\mathrm{PREP}_n|0^B\rangle
-
\sum_{k=0}^{K}\sqrt{\alpha_k}|k\rangle|g_k\rangle
\right\|
\leq 2^{-4n}.
\label{eq:poisson-prepare}
\end{equation}

\end{lemma}
Using the lemma we prepare the $B=O(n)$-qubit state
\begin{equation}
\mathrm{PREP}_n|0^B\rangle
\approx_{2^{-4n}}
\sum_{k=0}^{K}\sqrt{\alpha_k}|k\rangle|g_k\rangle.
\end{equation}
Next we append the state $|+^M\rangle$ and then use \Cref{lem:plus} to write
\begin{equation}
\mathrm{PREP}_n|0^B\rangle\otimes |+^M\rangle \approx_{2^{-4n+1}} \sum_{k=0}^{K}\sqrt{\alpha_k} |k\rangle|g_k\rangle|\phi_{k,M}\rangle.
\label{eq:sumLCU}
\end{equation}
Here we used the fact that $2^{-M/12}\leq 2^{-4n}$.
Then, we implement a diagonal unitary $D$ on the first register such that 
\begin{equation}
D|k\rangle=(-i)^k|k\rangle
\end{equation}
This operation only involves two bits of the first register; it can be implemented exactly using a constant number of gates in $\mathcal G$, giving
\begin{equation}
D\cdot\mathrm{PREP}_n|0^B\rangle\otimes |+^M\rangle \approx_{2^{-4n+1}} \sum_{k=0}^{K}(-i)^k\sqrt{\alpha_k} |k\rangle|g_k\rangle|\phi_{k,M}\rangle.
\label{eq:withi}
\end{equation}
Next we will need to implement a unitary $\mathrm{SELECT}$ that acts on $\lceil\log_2( K+1)\rceil+M+n$ qubits as
\begin{equation}
\mathrm{SELECT}|k\rangle |(\ell_1,\ell_2,\ldots, \ell_k)\rangle|+^{M-\sum \ell_j}\rangle |\phi\rangle=|k\rangle|(\ell_1,\ell_2,\ldots, \ell_k)\rangle|+^{M-\sum \ell_j}\rangle \prod_{j=1}^{k}P_{\ell_j}|\phi\rangle,
\label{eq:select}
\end{equation}
where $|\phi\rangle$ is any $n$-qubit state and 
\begin{equation}
P_\ell\equiv \begin{cases} Z_\ell & \ell \leq n\\
I & \text{ otherwise}.
\end{cases}
\label{eq:pl}
\end{equation}
We allow $\mathrm{SELECT}$ to do anything it likes if the input state is not of this form.
\begin{lemma}[SELECT lemma] Let  $K\leq M/4$ and $M=\Theta(n)$. A unitary $\mathrm{SELECT}$ acting as in \cref{eq:select} can be implemented exactly using $O(M)$ gates from $\mathcal{G}$ and ancillas.
\label{lem:select}
\end{lemma}
The proof of the Lemma is deferred to \Cref{sec:select}. Let us now adjoin an $n$-qubit data register initialized in $|+^n\rangle$ to our state in \cref{eq:withi} and then apply the $\mathrm{SELECT}$ operation (on all registers except the register holding $\ket{g_k}$) to get 
\begin{align}
&\mathrm{SELECT}\cdot  D\cdot \mathrm{PREP}_n|0^B\rangle|+^M\rangle|+^n\rangle \\
&\qquad \approx_{2^{-4n+1}}
\mathrm{SELECT}\sum_{k=0}^{K}(-i)^k\sqrt{\alpha_k} |k\rangle|g_k\rangle|\phi_{k,M}\rangle|+^n\rangle\\
&\qquad =\sum_{k=0}^{K}\sqrt{\alpha_k} |k\rangle|g_k\rangle
\sum_{\stackrel{\ell_1,\ell_2\ldots, \ell_k\in \mathbb{N}_+:}{\sum{\ell_j}\leq M}}
\frac{|(\ell_1,\ell_2,\ldots, \ell_k)\rangle |+^{M-\sum_{j} \ell_j}\rangle}{\sqrt{2^{\sum_{j}\ell_j}}}
\prod_{j=1}^{k}(-i P_{\ell_j})|+^n\rangle.
\end{align}
Define
\begin{equation}
\widetilde H_0:=\sum_{\ell=1}^{\infty}2^{-\ell}P_\ell
=\sum_{\ell=1}^{n}2^{-\ell}Z_\ell+2^{-n}I
=\frac{H}{\pi}+2^{-n}I.
\end{equation}
Next apply $\mathrm{PREP}_n^\dagger$ to the registers holding $\ket{k}\!\ket{g_k}$ and apply Hadamards to the $M$-qubit register. If we then project all qubits onto $|0\rangle$ except the final $n$-qubit register, we get
\begin{align}
(\langle 0^{B+M}|\otimes I_n)(I\otimes H^{\otimes M}\otimes I)\mathrm{PREP}^{\dagger}_n\cdot &\mathrm{SELECT}\cdot D\cdot \mathrm{PREP}_n|0^B\rangle|+^M\rangle|+^n\rangle\label{eq:prepselect}\\
&\approx_{2^{-4n+2}} \frac1{\Lambda_K}\sum_{k=0}^{K} \frac{\pi^k}{k!}
\sum_{\stackrel{\ell_1,\ldots,\ell_k\in\mathbb N_+:}{\sum_j\ell_j\leq M}}
\prod_{j=1}^{k}(-i2^{-\ell_j}P_{\ell_j})|+^n\rangle \nonumber \\
&\approx_{2^{-4n+3}}
\frac1{\Lambda_K}\sum_{k=0}^{K}\frac{(-i\pi\widetilde H_0)^k}{k!}|+^n\rangle, 
\label{eq:almost}
\end{align}
where in the last line we used that, for every $k\leq K\leq M/4$,
\begin{equation}
\sum_{\stackrel{\ell_1,\ldots,\ell_k\in\mathbb N_+:}{\sum_j\ell_j>M}}
2^{-(\ell_1+\cdots+\ell_k)}
=\frac1{2^M}\sum_{r=0}^{k-1} \binom{M}{r}
\leq2^{-M/6}\leq 2^{-4n}.
\end{equation}
Moreover, by choosing the constant $c$ in the definition of $K$ sufficiently large, we can ensure that $
\sum_{k>K}\frac{\pi^k}{k!}\leq2^{-4n}$ and therefore \cref{eq:almost} is $2^{-4n+4}$-close to
\begin{equation}
\frac1{\Lambda_K}e^{-i\pi\widetilde H_0}|+^n\rangle
= \gamma|\mathrm{grad}_n\rangle,
\end{equation}
where 
\begin{equation}\label{eq:gamma_n}
    \gamma=
\frac{e^{-i\pi2^{-n}}}{\Lambda_K} \quad  \text{and recall from \cref{eq:kc} that }\Lambda_K=\sum_{k=0}^{K}\frac{\pi^k}{k!} . 
\end{equation}
This means that~\cref{eq:prepselect} is $2^{-4n+4}\leq 2^{-2n}$-close to the above unnormalized state.

Note that 
since $1\leq \Lambda_K\leq e^\pi$, we have $|\gamma|\geq e^{-\pi}> 0.01$ as claimed.  Taking 
\begin{equation}
W=(I\otimes H^{\otimes M}\otimes I)\mathrm{PREP}_n^{\dagger}\cdot \mathrm{SELECT}\cdot D\cdot \mathrm{PREP}_n(I\otimes H^{\otimes (M+n)}),
\end{equation}
$\gamma$ as described above, and $a=M+B=O(n)$ completes the proof.
\end{proof}

\section*{Acknowledgments}
Various versions of ChatGPT, Claude, and Gemini were used to develop proof ideas, assist with proving theorems, and proofread drafts. In particular, LLMs played a significant role in the proof of \Cref{thm:optimalphase}. 

\bibliographystyle{alphaurl}
\bibliography{paper1}

\appendix

\section{Optimality}

\begin{proposition}\label{prop:optimal}
    Let $f$ be a super-polynomial function (i.e., $f(n) = n^{\omega(1)}$). It is not possible to prepare every $n$-qubit product state to trace distance $1/f(n)$ using an adaptive circuit with $O(n)$ gates over $\calG$.
\end{proposition}

\begin{proof}

We use a counting argument along the lines of the proof of~\cite[Claim~4.5]{GKW26}.

Set $\epsilon=1/f(n)$ and consider the product states $\ket{\psi_{\boldsymbol{\theta}}} =\bigotimes_{j=1}^n R(\theta_j)\ket{+}$, where each $\theta_j$ is chosen from a grid in $[0,1]$ with spacing $20\sqrt{\epsilon}$. For sufficiently large $n$, this gives $f(n)^{\Omega(n)}$ states with pairwise trace distance greater than $2\sqrt{\epsilon}$.

Suppose each of these states has an adaptive $\epsilon$-approximation with worst-case gate count $O(n)$. The adaptive circuit that approximately prepares  a given target state $\ket{\psi_{\boldsymbol{\theta}}}$  implements a quantum channel on $N=O(n)$ qubits
\begin{equation}
\mathcal{E}_{\boldsymbol{\theta}}(\proj{0^N})=\sum_{s} p_s 
\proj{\phi_s}\qquad
\ket{\phi_s}=\frac{1}{\sqrt{p_s}} M_s \ket{0^N},
\end{equation}
where the sum is over measurement outcomes $s$ that appear in the course of the adaptive circuit, $p_s=\|M_s|0^N\rangle\|^2$ is the probability of measurement outcomes $s$, and  $M_s$ is a sequence of $O(n)$ operators from $\mathcal{G}\cup \{\proj{0},\proj{1}\}$. Here we have measured any discarded ancillas to ensure that the state corresponding to each outcome $s$ is pure.

Since the adaptive circuit $\epsilon$-approximately prepares the given state in trace distance, we have
\begin{equation}
\mathrm{Tr}\left(\mathcal{E}_{\boldsymbol{\theta}}(\proj{0^N})  \cdot \proj{\psi_{\boldsymbol{\theta}}}\otimes I^{N-n}\right)\geq 1-\epsilon.
\end{equation}
This in turn implies that there exists some sequence of measurement outcomes $s=s(\boldsymbol{\theta})$ such that
\begin{equation}
\langle \phi_s|\cdot (\proj{\psi_{\boldsymbol{\theta}}}\otimes I^{N-n}) \cdot |\phi_s\rangle\geq 1-\epsilon,
\label{eq:ms}
\end{equation}
which implies (via Fuchs-van de Graaf) that the trace distance between $|\psi_{\boldsymbol{\theta}}\rangle$ and the corresponding reduced state of $|\phi_s\rangle$ is at most $\sqrt{\epsilon}$. 
Since the states $\{|\psi_{\boldsymbol{\theta}}\rangle\}$ are pairwise separated by greater than $2\sqrt{\epsilon}$ in trace distance, a fixed $|\phi_s\rangle$  can be within $\sqrt{\epsilon}$ of at most one target $|\psi_{\boldsymbol{\theta}}\rangle$. From this we infer that there must exist $f(n)^{\Omega(n)}$ different operators $M_s$, but this is a contradiction: there are at most $2^{O(n\log(n))}\ll f(n)^{\Omega(n)} $ possible operators $M_s$ that are built as products of $O(n)$ elements from $\mathcal{G}\cup \{\proj{0},\proj{1}\}$.
\end{proof}

The same bound holds for expected gate count: by Markov's inequality, branches costing at most twice the expectation have total probability at least $1/2$. Their conditional average infidelity is at most $2\epsilon$, from which we infer a branch $|\phi_s\rangle$  with trace distance at most $\sqrt{2\epsilon}$ and we can use an almost identical counting argument.

\begin{proposition}\label{prop:dyadic}
    For every sufficiently small $\eps>0$ there is a dyadic gate  $R(\pi/2^j)$ for some positive integer $j$, such that approximating it to error $\eps$ requires $\Omega(\log(1/\eps))$ $T$ gates even with unlimited ancilla, unlimited Cliffords and adaptive circuits.
\end{proposition}

\begin{proof}
By~\cite[Theorem~5.2]{beverland2020lower} with $C=2$, if $0<\delta<1/512$ and a single-qubit unitary $U$ satisfies
\begin{equation}
2\sqrt{2\delta}\leq |\langle0|U|1\rangle|^2\leq 6\sqrt{2\delta},
\label{eq:bchk-condition}
\end{equation}
then approximating $U$ to diamond-norm error $\delta$ requires $\Omega(\log(1/\delta))$ $T$ gates in expectation, even with arbitrary stabilizer ancillas, Clifford gates, and adaptive measurements.

To apply this bound to dyadic rotations, let $\theta_j=\pi/2^j$, $U_j=HR(\theta_j)H$, and $\delta_j=\sin^4(\theta_j/2)/32$. Then
\begin{equation}
|\langle0|U_j|1\rangle|^2
=\left|\frac{1-e^{i\theta_j}}{2}\right|^2
=\sin^2(\theta_j/2)
=4\sqrt{2\delta_j},
\end{equation}
so \cref{eq:bchk-condition} holds. Moreover, $\delta_j$ decreases to zero as $j\rightarrow \infty$ and
\begin{equation}
\frac{\delta_j}{\delta_{j+1}}=16\cos^4(\theta_j/4)\leq16.
\end{equation}
Thus, for every sufficiently small $\eps>0$, we can choose $j$ with $2\eps\leq\delta_j<32\eps$.

An implementation of $R(\theta_j)$ to diamond distance $\eps$, conjugated by free $H$ gates, approximates $U_j$ to diamond-norm error at most $2\eps\leq\delta_j$ with the same $T$-count. The lower bound is therefore $\Omega(\log(1/\delta_j))=\Omega(\log(1/\eps))$, as claimed.
\end{proof}

\section{Phase-halving gadget\label{sec:phasehalving}}

\begin{proposition}\label{prop:phasehalving}
For $\theta\in\mathbb{R}$ and single-qubit states $|\psi\rangle$ and $|\phi\rangle$, the phase-halving circuit on the right of \Cref{fig:rotation-resource-conversions} implements
\begin{equation}
R(\theta)|+\rangle\otimes|\psi\rangle\otimes|\phi\rangle\otimes|0\rangle 
\longmapsto R(\theta)|+\rangle\otimes R(\theta)|\psi\rangle\otimes R(\theta)|\phi\rangle\otimes|0\rangle.
\end{equation}
\end{proposition}

\begin{proof}
By linearity, it suffices to take $|\psi\rangle=|x\rangle$ and $|\phi\rangle=|y\rangle$ for bits $x,y\in\{0,1\}$. The four gates following $R(2\theta)$ undo the four immediately preceding it, restoring the data bits to $x,y$ and the ancilla to $|0\rangle$. We therefore track the catalyst and the phase acquired in each of three cases. On each computational-basis component, let $c$ denote the catalyst bit after the first two CNOTs.

\emph{Case $(x,y)=(0,0)$.} The first two CNOTs leave the catalyst unchanged. The next two CNOTs make both data bits equal to $c$, so the Toffoli writes $c$ into the ancilla and the subsequent CNOT changes it to $c\oplus c=0$. Thus $R(2\theta)$ contributes no phase.

\emph{Case $(x,y)=(0,1)$ or $(1,0)$.} Since exactly one of $x,y$ is $1$, the combined action of the first two CNOTs on the catalyst is $X$. The next two CNOTs change the data bits to $x\oplus c$ and $y\oplus c$. Since $x\ne y$, exactly one of these bits is $1$, so the Toffoli does nothing and the subsequent CNOT copies $c$ into the ancilla. Applying $R(2\theta)$ to the ancilla therefore has the same effect as applying it to the catalyst. The resulting catalyst state is
\begin{equation}
R(2\theta)X R(\theta)|+\rangle
=\frac{e^{\ii\theta}|0\rangle+e^{2\ii\theta}|1\rangle}{\sqrt2}
=e^{\ii\theta}R(\theta)|+\rangle.
\end{equation}
Thus the catalyst is restored and this input acquires phase $e^{\ii\theta}$.

\emph{Case $(x,y)=(1,1)$.} The first two CNOTs each apply $X$ to the catalyst, which cancel. The next two CNOTs make both data bits equal to $c\oplus 1$, so the Toffoli writes $c\oplus 1$ into the ancilla and the subsequent CNOT changes it to $(c\oplus 1)\oplus c=1$. Thus $R(2\theta)$ contributes phase $e^{2\ii\theta}$.

In all three cases, the catalyst is returned unchanged and the data input acquires phase $e^{\ii\theta(x+y)}$, which is exactly the action of $R(\theta)\otimes R(\theta)$.
\end{proof}

With free Clifford gates and adaptivity, the phase-halving gadget can be implemented using four $T$ gates in addition to one $R(2\theta)$ gate. The first Toffoli computes the AND of its two controls into a $|0\rangle$ ancilla, which costs four $T$ gates using the temporary logical-AND construction~\cite[Fig.~3]{gidney2018halving}. The two CNOTs surrounding $R(2\theta)$ restore this AND value, while the controls are unchanged. Thus the second Toffoli can be replaced by measuring the ancilla in the $X$ basis, applying $\mathrm{CZ}$ to the two controls if the outcome is $-1$, and resetting the ancilla to $|0\rangle$. This uncomputation uses only Clifford gates and a measurement, so it requires no additional $T$ gates.

\section{PREPARE and SELECT Lemmas}

\subsection{SELECT lemma\label{sec:select}}
\begin{proof}[Proof of \Cref{lem:select}]
Let $n$ be a positive integer, $M=\Theta(n)$, and let $s\in \{0,1\}^M$ and $x\in \{0,1\}^n$ be binary strings. Letting $t=|s|$ be the Hamming weight of $s$, we can write
\begin{equation}
s=0^{\ell_1-1}10^{\ell_2-1}1\ldots 0^{\ell_t-1}10^r
\label{eq:s}
\end{equation}
for some $0\leq t\leq M$, nonnegative integer $r$, and positive integers $\ell_1,\ell_2,\ldots \ell_t$ such that $\sum_{j=1}^{t} \ell_j\leq M$.

In the following we implement an operation $\mathrm{SELECT}$ such that
\begin{equation}
\mathrm{SELECT}|k\rangle|s\rangle|x\rangle=|k\rangle|s\rangle\prod_{j=1}^{\min\{t,k\}}P_{\ell_j}|x\rangle.
\label{eq:implement}
\end{equation}
where $s$ is given by \cref{eq:s} and $P_{\ell}$ is given by \cref{eq:pl}.
This is sufficient to implement the desired functionality \cref{eq:select}. 

The key ingredients in our implementation are (a) computing a tree-like data structure from $s$ which hierarchically organizes the information about the ``gaps'' $\ell_j$  according to their size, and (b) fast classical circuits for \textit{multiselection} due to Holmgren and Rothblum \cite{HolmgrenRothblum24}:

\begin{theorem}[{\cite{HolmgrenRothblum24}}]\label{thm:hr}
There is a bounded fan-in Boolean
circuit for the multiselection map
\begin{equation}
\operatorname{Sel}_{n\to m}:\B^n\times [n]^m\to \B^m,\qquad
(x,i_1,\ldots,i_m)\mapsto (x_{i_1},\ldots,x_{i_m}),
\end{equation}
of size $O(n+m\log^3 n)$ and depth $O(\log(n+m))$.  
\end{theorem}

The implementation of $\mathrm{SELECT}$ first rearranges the data from the input string $s$ into a certain tree-like data structure that we now describe. 

Let $R$ be the smallest integer such that $M+1\leq 2^R$. Consider a complete binary tree $T$ which has $2^R$ leaves. The height of a node in the tree is defined to be the number of edges in the shortest path to a leaf node (i.e., leaves have height zero and the root has height $R$).

Given the input string in \cref{eq:s}, we shall think of the $t$ ``gaps'' $\{\ell_j\}_j\in [t]$ as being assigned to $t$ distinct nodes in the tree. To do this, first assign each bit of the $M+1$ bit string $1s$ left-to-right to the leftmost $M+1$ leaves of the tree. Then for each $j\in [t]$, the gap $\ell_j-1$ is the number of zero-leaves between the $j$th one-leaf and $j+1$st one leaf (again reading from left to right). We assign the gap $\ell_j$ to the least common ancestor of the $j$th one-leaf and the $j+1$st one leaf.  This is depicted in~\Cref{fig:tree}. 

\begin{figure}[htbp!]
\centering
\begin{tikzpicture}[scale=1.1,
  x=0.68cm,y=0.82cm,
  edge/.style={black!55,line width=0.45pt},
  treenode/.style={circle,draw=black,fill=white,minimum size=6.2mm,inner sep=0pt,font=\scriptsize},
  oneleaf/.style={circle,draw=blue!70!black,fill=blue!10,very thick,minimum size=6.2mm,inner sep=0pt,font=\scriptsize},
  selectedowner/.style={circle,draw=green!40!black,fill=green!18,very thick,minimum size=6.2mm,inner sep=0pt},
  owner/.style={circle,draw=black!60,fill=black!8,very thick,dashed,minimum size=6.2mm,inner sep=0pt},
  lab/.style={font=\scriptsize,inner sep=1pt},
  gapbrace/.style={decorate,decoration={brace,mirror,amplitude=3.5pt},thick,black!65}
]

\foreach \i/\bit in {0/1,1/0,2/1,3/0,4/0,5/1,6/0,7/0,
                     8/0,9/1,10/0,11/0,12/0,13/0,14/0,15/1} {
  \node[treenode] (leaf\i) at (\i,0) {\bit};
}

\foreach \i/\bit in {0/1,2/1,5/1,9/1,15/1} {
  \node[oneleaf] (leaf\i) at (\i,0) {\bit};
}

\foreach \i/\x in {0/0.5,1/2.5,2/4.5,3/6.5,
                   4/8.5,5/10.5,6/12.5,7/14.5} {
  \node[treenode] (a\i) at (\x,1.15) {};
}
\foreach \i/\x in {0/1.5,1/5.5,2/9.5,3/13.5} {
  \node[treenode] (b\i) at (\x,2.30) {};
}
\foreach \i/\x in {0/3.5,1/11.5} {
  \node[treenode] (c\i) at (\x,3.45) {};
}
\node[treenode] (d0) at (7.5,4.60) {};

\node[selectedowner,label={[lab]right:owns $\ell_1=2$}] (b0) at (1.5,2.30) {};
\node[selectedowner,label={[lab]left:owns $\ell_2=3$}] (c0) at (3.5,3.45) {};
\node[selectedowner,label={[lab]above:owns $\ell_3=4$}] (d0) at (7.5,4.60) {};
\node[owner,label={[lab]right:owns $\ell_4=6$}] (c1) at (11.5,3.45) {};

\begin{scope}[on background layer]
\draw[edge] (a0.south west) -- (leaf0.north);
\draw[edge] (a0.south east) -- (leaf1.north);
\draw[edge] (a1.south west) -- (leaf2.north);
\draw[edge] (a1.south east) -- (leaf3.north);
\draw[edge] (a2.south west) -- (leaf4.north);
\draw[edge] (a2.south east) -- (leaf5.north);
\draw[edge] (a3.south west) -- (leaf6.north);
\draw[edge] (a3.south east) -- (leaf7.north);
\draw[edge] (a4.south west) -- (leaf8.north);
\draw[edge] (a4.south east) -- (leaf9.north);
\draw[edge] (a5.south west) -- (leaf10.north);
\draw[edge] (a5.south east) -- (leaf11.north);
\draw[edge] (a6.south west) -- (leaf12.north);
\draw[edge] (a6.south east) -- (leaf13.north);
\draw[edge] (a7.south west) -- (leaf14.north);
\draw[edge] (a7.south east) -- (leaf15.north);

\draw[edge] (b0.south west) -- (a0.north);
\draw[edge] (b0.south east) -- (a1.north);
\draw[edge] (b1.south west) -- (a2.north);
\draw[edge] (b1.south east) -- (a3.north);
\draw[edge] (b2.south west) -- (a4.north);
\draw[edge] (b2.south east) -- (a5.north);
\draw[edge] (b3.south west) -- (a6.north);
\draw[edge] (b3.south east) -- (a7.north);

\draw[edge] (c0.south west) -- (b0.north);
\draw[edge] (c0.south east) -- (b1.north);
\draw[edge] (c1.south west) -- (b2.north);
\draw[edge] (c1.south east) -- (b3.north);

\draw[edge] (d0.south west) -- (c0.north);
\draw[edge] (d0.south east) -- (c1.north);
\end{scope}

\draw[gapbrace] ($(leaf0.south)+(0,-0.35)$) -- ($(leaf2.south)+(0,-0.35)$)
node[midway,below=3.5pt,font=\scriptsize] {$\ell_1$};
\draw[gapbrace] ($(leaf2.south)+(0,-0.35)$) -- ($(leaf5.south)+(0,-0.35)$)
  node[midway,below=3.5pt,font=\scriptsize] {$\ell_2$};
\draw[gapbrace] ($(leaf5.south)+(0,-0.35)$) -- ($(leaf9.south)+(0,-0.35)$)
  node[midway,below=3.5pt,font=\scriptsize] {$\ell_3$};
\draw[gapbrace] ($(leaf9.south)+(0,-0.35)$) -- ($(leaf15.south)+(0,-0.35)$)
  node[midway,below=3.5pt,font=\scriptsize] {$\ell_4$};

\end{tikzpicture}
\caption{Example of the tree data structure with $M=15$, $s=010010001000001=(2,3,4,6)$, and $k=3$. The green nodes are the owners of the gaps $\ell_j$ with $j\leq k$.\label{fig:tree}}
\end{figure}

Two important properties of this construction are as follows:
\begin{enumerate}
\item{ Each of the gaps $\{\ell_j\}_{j\in [t]}$ is assigned to a unique node.}
\item{If a gap $\ell_j$ is assigned to a node at height $h$, then $\ell_j< 2^h$.}
\end{enumerate}
Note that, to implement the desired operation \cref{eq:implement}, we only need to consider the gaps $\ell_j$ that have $j\leq k$ and $\ell_j\leq n$.

Operationally, to construct this data structure we shall compute a bit string $I_v$ for each node $v\in T$.  Letting $h=h(v)$ be the height of the node, we have $I_v\in \{0,1\}^{h}$ defined by
\begin{equation}
I_v=\begin{cases}
\ell_j & \text{ if } v \text{ owns } \ell_j \text{ with } j\leq k \text{ and } \ell_j\leq n.\\
0^h & \text{ otherwise.}
\end{cases}
\label{eq:iv}
\end{equation}
Note that in the first case $\ell_j$ is represented in binary as an $h$-bit string. This is always possible due to property 2. listed above. We shall write $T(k,s,n)$ for the bitstring that consists of the concatenation of $\{I_v\}_{v\in T}$ in some chosen order. Note that this bit string has length linear in $n$. Indeed, since there are $2^R/2^h$ nodes at height $h$, we have
\begin{equation}
\mathrm{length}(T(k,s,n))=\sum_{h=0}^{R} \frac{2^Rh}{2^h}\leq \sum_{h=0}^{\infty} \frac{2^Rh}{2^h} =2^{R+1}=O(n).
\end{equation}
We now show it can be computed reversibly in linear time.
\begin{lemma}
For each positive integer $n$, the map
\begin{equation}
\mathcal{T}_n: |k\rangle|s\rangle|0\rangle \mapsto |k\rangle |s\rangle |T(k,s,n)\rangle
\end{equation}
can be implemented by a circuit with $O(n)$ elementary gates from $\mathcal{G}$ and $O(n)$ ancillas. \label{lem:tree}
\end{lemma}

\begin{proof}
Write $N=2^R$ for the number of leaves in the tree and let
\begin{equation}
a=1s0^{N-M-1}
\end{equation}
be the $N$-bit string obtained from $s$ by adding a $1$ at the beginning and padding with zeros at the end.

In the following we use the fact that comparison, addition, and subtraction of integers represented by $w$-bit strings can be performed exactly by reversible classical  circuits using $O(w)$ Toffoli gates, NOT gates, and ancillas. The same holds if we compile over $\mathcal{G}$ since Toffoli and NOT can be implemented with $O(1)$ gates from $\mathcal{G}$.

For each node $v$ in the tree, let $L(v)$ be the set of all leaves that are below $v$. Note that if $v$ is at height $h$ then $|L(v)|=2^h$.

In order to compute $T(k,s,n)$ we will need to first compute some auxiliary data. This is done in two passes through the nodes of the tree, a bottom-up pass followed by a top-down pass.

 In the bottom-up pass we compute, for each node $v$, the following data:
\begin{align}
c_v&=|\{i\in L(v): a_i=1\}|\\
f_v&= \text{ Index of the leftmost $1$-leaf in $L(v)$}\\
r_v&=\text{ Index of the rightmost $1$-leaf in $L(v)$}.
\end{align}
Here the index of a leaf in $L(v)$ is an integer between $1$ (the leftmost leaf in $L(v)$) and $2^h$ (rightmost leaf in $L(v)$), where $h$ is the height of $v$. We also adopt the convention that $f_v=r_v=0$ if $c_v=0$. Moreover, $0\leq c_v\leq 2^h$. So we can represent $c_v,f_v,r_v$ as $(h+1)$-bit strings. If $v$ has left and right children $u,w$ respectively then
\begin{align}
c_v=c_u+c_w \qquad f_v=\begin{cases}
f_u & c_u>0\\
2^{h-1}+f_w & c_u=0 \text{ and } c_w>0\\
0 & c_u=c_w=0.
\end{cases}
\qquad
r_v=\begin{cases}
2^{h-1}+r_w & c_w>0\\
r_u & c_w=0.
\end{cases}
\label{eq:bottomup}
\end{align}

From \cref{eq:bottomup} we see that  computing $c_v,f_v, r_v$  given this data for its children only requires a constant number of additions and comparisons of $h+1$-bit strings. Since there are $N/2^h$ nodes of height $h$, the total number of elementary gates needed to compute this data for all nodes is
\begin{equation}
\sum_{h=0}^{R} \frac{N}{2^h}O(h+1)=O(N)=O(n).
\label{eq:sumh}
\end{equation}
This is also an upper bound on the total number of ancillas used.

In the top-down pass we compute a ``1-leaf budget'' for each node $v$ that is not a leaf.  Let $p_v$ be the number of $1$-leaves that are strictly to the left of $L(v)$; these $1$-leaves have already been spent out of our total budget of $k$. The remaining $1$-leaf budget for node $v$, is then
\begin{equation}
q_v\equiv \min\{(k-p_v)_{+},2^h\}
\label{eq:qvdef}
\end{equation}
where $i_{+}=\max\{i,0\}$ for an integer $i$. We think of $q_v$ as the number of $1$-leaves that we care about (up to $k$ total) that could in principle appear in $L(v)$. The minimum is taken with $2^h$ since this is the total size of $L(v)$.

The root node has $q_v=\min\{k,N\}$; if $v$ has left child $u$ and right child $w$ then
\begin{equation}
q_u=\min\{q_v,2^{h-1}\} \qquad q_w=\min\{(q_v-c_u)_{+},2^{h-1}\}.
\end{equation}
This allows us to compute $q_u,q_w$ from $q_v,c_u$ using $O(h)$ elementary gates. The total cost of the top-down pass is then upper bounded as in \cref{eq:sumh}.

Now suppose $v$ is a node with left child $u$ and right child $w$. Then:
\begin{itemize}
\item{node $v$  owns a gap iff $c_u>0$ and $c_w>0$.}
\item{If $v$ owns a gap, its length is $d_v=(2^{h-1}+f_w)-r_u$.}
\item{If $v$ owns a gap, it is one of the first $k$ gaps iff $c_u\leq q_v$ (to see this note that the gap owned by $v$ is the $j$th gap, where $j=p_v+c_u$, and then use \cref{eq:qvdef}).}
\end{itemize}
Using these properties and the data $c_v,f_v,r_v,q_v$ computed from the two passes, we can now compute the data $I_v$ from \cref{eq:iv} for each non-leaf node $v$ (note that leaf nodes never own gaps so $I_v=0$ for them). In particular, 
\begin{equation}
I_v=\begin{cases}
d_v & c_u>0 \text{ and } c_u \leq q_v \text{ and }c_w>0\text{ and }d_v\leq n.\\
0^h & \text{otherwise}.
\end{cases}
\end{equation}
Since this uses $O(h)$ elementary gates per node, the total cost is $O(N)=O(n)$ by \cref{eq:sumh}. After computing $I_v$ for all nodes we then uncompute all of the auxiliary data and return all ancillas to the all-zero state.
\end{proof}

Once we have computed $T(k,s,n)$, the algorithm proceeds by using HR multiselection (\Cref{thm:hr}) to implement all unitaries $P_\ell$ corresponding to gaps $\ell$ that are owned by vertices at a given height $h$, for each  $h=0,1,\ldots,R$ one by one. This is described in \Cref{alg:select}. The runtime of this algorithm is $O(n)$ which can be seen as follows:
\begin{itemize}
\item Computing the tree data structure in line 2 takes $O(n)$ time due to \Cref{lem:tree}.
\item The runtime of HR multiselection in line 4 is $O(2^h+2^R/2^h h^3)$ since there are $2^R/2^h$ nodes at height $h$ and $I_v< 2^h$. Summing this cost over all heights gives
\begin{equation}
\sum_{h=0}^{R}O(2^h+2^R/2^h h^3)=O(2^R)=O(n).
\end{equation}
\end{itemize}

\Cref{alg:select} performs the operation
\begin{equation}
|k\rangle|s\rangle|x\rangle \rightarrow |k\rangle|s\rangle \bigotimes_{v: I_v\neq 0} Z_{I_v}|x\rangle=|k\rangle|s\rangle \bigotimes_{j\in [k]: \ell_j\leq n} Z_{\ell_j}|x\rangle,
\end{equation}
where the last equality is due to \cref{eq:iv}. Thus we have implemented the desired operation \cref{eq:implement}.

\begin{algorithm2e}[htbp!]
\DontPrintSemicolon
\caption{\texttt{SELECT}}
\label{alg:select}
\textbf{Input: }$\ket{k,s,x}$ \;
Compute the tree data structure $T(k,s,n) = \{ I_v: \ v \text{ is a tree node}\}$ into an ancilla register\;

\For{$h = 1 \to R$}
{
Consider all nodes $v$ of height $h$.
Use Holmgren--Rothblum multiselection (\Cref{thm:hr}) to read all bits $x_{I_v}$ for all nodes $v$ at this height simultaneously:
\begin{equation*}
\ket{x} \otimes \left( \bigotimes_{v: \, \mathrm{height} \, h} \ket{I_v} \right) \otimes \ket{0} \mapsto \ket{x} \otimes \left( \bigotimes_{v: \, \mathrm{height} \, h} \ket{I_v} \right) \otimes \ket{\bigoplus_{v: \, \mathrm{height} \, h}x_{I_v}}
\end{equation*}
Note that at this height $I_v\in \{0,1\}^h$, and we use the convention $x_0=0$. \;

Apply a single-qubit Pauli $Z$ to the last parity ancilla.  The state becomes
\begin{equation*}
    \ket{x} \otimes \left( \bigotimes_{v: \, \mathrm{height} \, h} \ket{I_v} \right) \otimes (-1)^{\sum_{v: \, \mathrm{height} \, h}x_{I_v}}  \ket{\bigoplus_{v: \, \mathrm{height} \, h}x_{I_v}}
\end{equation*} \;

Uncompute the HR multiselection workspace for this height.\;
}

Uncompute the tree data structure workspace.
\end{algorithm2e}
\end{proof}

\subsection{PREPARE lemma \label{sec:prepare}}

\begin{proof}[Proof of \Cref{lem:poisson-prepare}]

Recall that our goal is to prepare a state that is $2^{-4n}$-close to
\begin{equation}
|\tilde{\beta}\rangle=\frac{1}{\sqrt{\Lambda_K}}\sum_{k=0}^{K} \sqrt{\frac{\pi^k}{k!}}|k\rangle|g_k\rangle,
\end{equation}
where $|g_k\rangle$ are some normalized states, $\Lambda_K=\sum_{k=0}^{K}\pi^k/k!$ is a normalizing constant, and
\begin{equation}
K=\left\lceil \frac{cn}{\log_2(n+2)}\right\rceil,
\end{equation}
where $c$ is a constant we may choose to be sufficiently large. To achieve this, it suffices to describe a (heralded) state preparation algorithm which succeeds in preparing a state that is $2^{-an}$-close to $|\tilde{\beta}\rangle$ with a probability that is lower bounded by a constant independent of $n$ and which we can compute explicitly to within $2^{-an}$ error, for some constant $a>4$. Then we use a constant number of rounds of exact amplitude amplification just as in the proof of \Cref{thm:optimalphase} (we do not repeat this argument here). In the remainder of the proof, we describe the heralded state preparation algorithm and establish that it has the above properties.

The factor of $\sqrt{\pi}^k$ makes this slightly nontrivial. The first step is to prepare a similar state where $\pi$ is replaced with $4$. In particular, we first describe a heralded algorithm which, with known constant success probability,  prepares a state
\begin{equation}
|\psi_0\rangle=\frac{1}{\sqrt{S_K}}\sum_{k=0}^{K} \sqrt{\frac{4^k}{k!}}|k\rangle|g_k\rangle \qquad \qquad S_K\equiv\sum_{k=0}^{K} 4^k/k!.
\label{eq:psi0}
\end{equation}

To do this, it suffices to construct a heralded classical algorithm that takes as input a string $r\in \{0,1\}^m$ of uniformly random bits, where $m=O(K\log K)$ and, conditioned on success, outputs an integer $k\in \{0,1,\ldots, K\}$ with probability $p(k)\propto 4^k/k!$.   Then we can implement the sampler on a quantum computer and run it coherently on a uniform superposition $|+^m\rangle$. Postselecting on a success flag then gives the state $|\psi_0\rangle$. Here $g_k$ consists of  ``garbage''
ancilla bits used in the sampling algorithm. 

We start with an algorithm $\calA_K$ that outputs $k$ with probability $\propto 1/k!$: 
\begin{lemma}
    For any $K\geq 1$, there is a classical sampling algorithm $\calA_K$ that outputs $k\in \{0,1,\ldots, K\}$ with probability $\eta_K/(3k!)$ or declares failure, where $\eta_K\geq 3/4$ is explicit and depends only on $K$.
    The algorithm uses $O(K \log K)$ uniformly random bits and elementary gates.
\end{lemma}
\begin{proof}
We first describe a simple algorithm that outputs $k$ with probability $1/(3k!)$. 
Draw $B$ uniformly from $\{0,1,2\}$. If $B=0$ or $1$, output $k=B$. If $B=2$, then for each $j\in \{2, \ldots, K\}$, draw $R_j$ uniformly from $\{0, \ldots, j-1\}$: continue if $R_j=0$; output $k=j$ if $R_j=1$; declare failure otherwise. 
If the loop finishes without an output, also declare failure. 

What remains is to generate $B, R_2, \ldots, R_K$ from uniformly random bits.  
For each task of sampling uniformly from $\{0,\ldots, d-1\}$ for some $d$, we can draw $\lceil\log_2 d\rceil$ fair bits and redraw if the resulting integer $\geq d$, so each draw succeeds with probability $d/2^{\lceil\log_2 d\rceil} > 1/2$. 
Let $T$ be the total number of draws needed to generate all $K$ integers, and thus $\E[T]< 2K$. 
We terminate and declare failure after $8K$ draws, which uses at most $O(K \log K)$ fair bits in total. 
By Markov's inequality, $\eta_K = \Pr[T \leq 8K] \geq 3/4$. 
Hence the probability of outputting $k$ is $\eta_K/(3k!)$.

To make $\eta_K$ explicit, note that $T = N_1+\cdots + N_K$ where $N_1, \ldots, N_K$ are independent geometric random variables for the numbers of draws to generate $B, R_2, \ldots, R_K$. 
The success probabilities are
\begin{equation}
  p_1 = 3/4 \qquad \text{ and } \qquad p_j = j/2^{\lceil \log_2 j\rceil} \text{ for each }2\leq j\leq K. 
\end{equation}
Since $\Pr[N_j = r] = p_j (1-p_j)^{r-1}$ for each $r\geq 1$, we have
\begin{equation}
  \eta_K=\Pr[T \leq 8K] = \paren*{\prod_{j=1}^K p_j} \sum_{\substack{r_1, \ldots, r_K\geq 1\\ r_1+\cdots +r_K\leq 8K}} \prod_{j=1}^K (1-p_j)^{r_j-1}. 
\end{equation}

The acceptance or rejection of each draw uses $O(\log K)$ Toffoli gates. Therefore, the total number of elementary gates and ancillas used in $\calA_K$ is $O(K\log K)$.
\end{proof}

Then we can sample $k$ with probability $\propto 4^k/k!$ as follows: 
Run $\calA_K$ four times independently and add their outputs. 
Declare failure if one of the four runs fails or the sum exceeds $K$. Then for $0\leq k\leq K$, 
\begin{equation}
  \Pr[\text{success and sum to }k] = \frac{\eta_K^4}{3^4} \sum_{k_1+k_2+k_3+k_4 = k} \frac{1}{k_1! k_2! k_3! k_4!} = \frac{\eta_K^4}{81}\cdot \frac{4^k}{k!}.
\end{equation}
The success probability is therefore $\eta_K^4 S_K/81 = \Omega(1)$. 
The addition of the four outputs can be coherently implemented as 
\begin{equation}
  \ket{k_1, k_2, k_3, k_4}\ket{0^b} \mapsto \underbrace{\ket{k_1, k_2, k_3, k_4}}_{\text{``garbage''}}\ket{k_1+ k_2+k_3+k_4} , 
\end{equation}
where $b = \lceil \log_2(4K+1)\rceil$, by a reversible ripple-carry adder using $O(b) = O(\log K)$ Toffoli and CNOT gates. 
Since each $\calA_K$ uses $O(K\log K)$ elementary gates and ancillas, the total gate and ancilla costs of this sampling step are also $O(K\log K)$. 

Starting from the state $|\psi_0\rangle$ from \cref{eq:psi0}, we now show how to (approximately) multiply the amplitude of $|k\rangle|g_k\rangle$ by the desired factor $(\sqrt{\pi}/2)^k$. This is a little bit tricky and to achieve this we will use  the SELECT lemma (\Cref{lem:select}).

Since $\frac{\sqrt{\pi}}{2}<1$, we can write the binary expansion
\begin{equation}
\frac{1-\sqrt{\pi}/2}{2}=\sum_{\ell=1}^{\infty}2^{-\ell} d_\ell
\end{equation}
where $d_\ell\in \{0,1\}.$ We shall work with the first $L=128n$ bits of this binary expansion. Starting from $|\psi_0\rangle$, we append two $L$-bit registers in the state $|+^L\rangle|d_1d_2\ldots d_L\rangle$, giving a state
\begin{align}
&\frac{1}{\sqrt{S_K}}\sum_{k=0}^{K}\sqrt{\frac{4^k}{k!}}|k\rangle|g_k\rangle|+^L\rangle|d_1d_2\ldots d_L\rangle \nonumber \\
&\qquad \approx_{2^{-L/12}}\frac{1}{\sqrt{S_K\cdot 2^L}}\sum_{k=0}^{K}\sum_{s\in \{0,1\}^L: |s|\geq k} \sqrt{\frac{4^k}{k!}}|k\rangle|g_k\rangle|s\rangle |d_1d_2\ldots d_L\rangle\\
&\qquad =\frac{1}{\sqrt{S_K}}\sum_{k=0}^{K} \sqrt{\frac{4^k}{k!}}|k\rangle|g_k\rangle \sum_{\stackrel{\ell_1,\ell_2,\ldots, \ell_k\in \mathbb{N}_{+}}{\sum {\ell_j}\leq L}} \frac{1}{\sqrt{2^{\sum{\ell_j}}}}|(\ell_1,\ell_2,\ldots, \ell_k)\rangle|+^{L-\sum_j\ell_j}\rangle|d_1d_2\ldots d_L\rangle.
\end{align}
Here the error bound follows from \Cref{lem:plus}.
Next we apply the SELECT operation with parameters $(K,M,n)=(K,L,L)$, which gives (up to the same error)
\begin{equation}
\frac{1}{\sqrt{S_K}}\sum_{k=0}^{K} \sqrt{\frac{4^k}{k!}}|k\rangle|g_k\rangle \sum_{\stackrel{\ell_1,\ell_2,\ldots, \ell_k\in \mathbb{N}_{+}}{\sum {\ell_j}\leq L}} \frac{1}{\sqrt{2^{\sum{\ell_j}}}}|(\ell_1,\ell_2,\ldots, \ell_k)\rangle|+^{L-\sum_j\ell_j}\rangle\left(\prod_{j=1}^{k} (1-2d_{\ell_j})\right)|d_1d_2\ldots d_L\rangle{\color{blue}.}
\end{equation}
Finally, measure the first $L$-bit ancilla register in the $X$-basis and declare success if you get $\ket{+^L}$. 
The unnormalized vector $|\psi_1\rangle$ corresponding to this successful outcome satisfies
\begin{equation}
|\psi_1\rangle\approx_{2^{-L/12}} \frac{1}{\sqrt{S_K}}\sum_{k=0}^{K}q_{k,L}\cdot \sqrt{\frac{4^k}{k!}}|k\rangle|g_k\rangle
\qquad \text{ where } \qquad q_{k,L}=\sum_{\stackrel{\ell_1,\ell_2,\ldots, \ell_k\in \mathbb{N}_{+}}{\sum {\ell_j}\leq L}} \frac{1}{2^{\sum{\ell_j}}}\left(\prod_{j=1}^{k} (1-2d_{\ell_j})\right).
\end{equation}
Note that
\begin{equation}
q_{k,\infty}=\left(\sum_{\ell\in \mathbb{N}_{+}} \frac{1}{2^\ell}(1-2d_{\ell})\right)^k=(\sqrt{\pi}/2)^k
\end{equation}
is the multiplicative factor we want to apply, and
\begin{align}
|q_{k,\infty}-q_{k,L}|
&\leq
\sum_{\substack{\ell_1,\ldots,\ell_k\geq1\\                  \ell_1+\cdots+\ell_k>L}}
2^{-(\ell_1+\cdots+\ell_k)}
=
2^{-L}\sum_{r=0}^{k-1} \binom{L}{r}
\leq
2^{-L/6},
\label{eq:qkl-error}
\end{align}
where we used $k\leq K\leq L/4$ and
$h(1/4)<5/6$ (and when $k=0$ we have $q_{k,L}=q_{k,\infty}=1$).

Now \cref{eq:qkl-error} gives
\begin{equation}
\left\|
|\psi_1\rangle
-
\sqrt{\frac{\Lambda_K}{S_K}}\,|\tilde{\beta}\rangle
\right\|
\leq 2^{-L/12}+
\left(\frac1{S_K}
\sum_{k=0}^{K}
\frac{4^k}{k!}
\left|q_{k,L}-q_{k,\infty}\right|^2\right)^{1/2}
\leq
2^{-L/12}+2^{-L/6}\leq 2^{-10n}.
\label{eq:psi1-error}
\end{equation}
Since $|\tilde{\beta}\rangle$ is normalized and
\begin{equation}
1\leq \Lambda_K\leq S_K\leq e^4, 
\label{eq:ratio}
\end{equation}
using Eq.~\eqref{eq:psi1-error} we see that the normalized output state $\ket{\psi_1}/\norm{\ket{\psi_1}}$ of this procedure is (say) $2^{-6n}$-close to the target state $|\tilde{\beta}\rangle$. Moreover, the success probability of this second step is
\begin{equation}
\frac{\Lambda_K}{S_K}+O(2^{-10n}),
\end{equation}
which is bounded below by a positive constant due to Eq.~\eqref{eq:ratio}.
\end{proof}

\end{document}